\documentclass[11pt]{article}
\usepackage{times}
\usepackage[pdftex]{graphicx}
\usepackage{amssymb}
\usepackage{amsmath}
\usepackage{amsthm}
\usepackage{mdframed}
\usepackage{algorithm}
\usepackage{mathtools}
\usepackage{comment}
\usepackage{cleveref}
\usepackage{fullpage}
\newcommand*\samethanks[1][\value{footnote}]{\footnotemark[#1]}

\newcommand{\C}{\mathbb{C}}

\renewcommand{\cref}{\Cref}

\newcommand{\OPT}{\mathsf{OPT}}
\newcommand{\kk}{k_{\scriptscriptstyle \OPT}}
\newcommand{\kkk}{k'}
\newcommand{\kkkk}{k_{\scriptscriptstyle \OPT}'}

\DeclarePairedDelimiter\floor{\lfloor}{\rfloor}
\DeclarePairedDelimiter\ceil{\lceil}{\rceil}
\newtheorem{theorem}{Theorem}
\newtheorem{lemma}{Lemma}

\begin{document}
\title{Graph Balancing with Two Edge Types }
\author{Deeparnab Chakrabarty\thanks{Microsoft Research, 9 Lavelle Road, Bangalore, India 560001. {\tt deeparnab, kirankumar.shiragur @gmail.com}} \and Kirankumar Shiragur\samethanks}
\date{}
\maketitle
\begin{abstract}
	In the graph balancing problem the goal is to orient a weighted undirected graph to minimize the maximum weighted in-degree. This special case of makespan minimization is NP-hard to approximate to a factor better than $3/2$ even when there are only two types of edge weights. In this note we describe a simple $3/2$ approximation for the graph balancing problem with two-edge types, settling this very special case of makespan minimization.
\end{abstract}
\section{Introduction}
In the graph balancing problem, we are given an undirected graph $G=(V,E)$ with weights $p_e$ on edge $e$. The graph could have parallel edges and self loops. The goal is to find an orientation of the edges so as to minimize the maximum weighted in-degree. 

This problem is a special case of makespan minimization, a classic problem in approximation algorithms, where
the input is a collection of jobs $J$ and machines $M$, and processing time $p_{ij}$ for $i\in M$ and $j\in J$.
The goal is to find an allocation of all jobs to machines so as to minimize $\max_{i\in M} \sum_{j\in S_i} p_{ij}$. Note that graph balancing is a special case 
when nodes correspond to machines and each edge corresponds to a job which has processing time $p_e$ on each of the endpoints and $\infty$ on all other machines. In fact, this is a special case of the so-called restricted assignment machine scheduling problem where each job has an inherent processing time but can only be assigned to a subset of the machines.

In 1990, Lenstra, Shmoys, and Tardos~\cite{LST} described a now classic $2$-approximation for this problem, and also proved it is NP hard to obtain an approximation factor better than $3/2$. Closing this gap has challenged many researchers in the past three decades. The graph balancing problem was introduced in 2008 by  Ebenlendr, Krcal and Sgall~\cite{EKS} as an interesting special case of makespan minimization, and they gave an $1.75$-approximation algorithm for this problem. \cite{EKS} also showed that even in the graph balancing problem, it is NP-hard to obtain an approximation factor better than $3/2$. 
%
It is perhaps testament to the difficulty of the makespan minimziation problem that we do not know the ``true answer'' even for graph balancing. 
In 2011, Kolliopoulos and Moysoglou~\cite{KM} looked at an even more special case -- where there are only two types of edge weights! It turns out that even in this special, special case, one can't get better than $3/2$-approximation; the reduction of Ebenlendr et al~\cite{EKS} uses only two distinct edge weights. Kolliopoulos and Moysoglou~\cite{KM} give a $1.652$-factor approximation for this case. 

The purpose of this note is to show that a slight modification of techniques of ~\cite{KM}  in fact gives an optimal $1.5$-approximation for graph balancing with two types of edge weights.

\section{Algorithm}
By scaling, we may assume that the edge weights are either $1$ or $c < 1$ for some real number $c$. 
In fact, Kolliopoulos and Moysoglu~\cite{KM} showed if $c = 1/k$ for some integer $k$, then there ia a $1.5$-approximation algorithm. 
Our main observation is with a bit of case analysis, we can generalize to arbitrary $c$.

Fix an instance of the graph balancing problem and let $\OPT$ be its optimum value. The following lemma implies we may assume $\OPT < 2$.
%
%
\begin{lemma}\label{lem:1}
	There is a $\frac{3}{2}$-approximation algorithm if $\OPT \geq 2$. 
\end{lemma}
\begin{proof}
	The algorithm in~\cite{LST} returns an orientation with weighted in-degree $\le \OPT+1\le \frac{3}{2}\OPT$ if $\OPT\geq 2$. 
\end{proof}
Therefore, we may assume $\OPT = 1 + \kk c$ or $OPT=\kkkk c$ for some non-negative integers $\kk < 1/c$ and $\kkkk < 2/c$. 
%
\noindent
Let $k:=\floor{\frac{1}{c}}$. Given some integers $p,q$, we describe a flow networks $N(p,q)$ very similar to that defined by~\cite{KM} differing only in the capacities. For each arc in the network we define a flow lower bound and upper bound; a flow is feasible if the flow through all arcs lies within their range. The lower bound unless explicitly mentioned is $0$, and the upper bound unless explicitly mentioned is $\infty$.
Call an edge $e$ small if $p_e = c$, and big otherwise. \\

\noindent
{\bf Network Description $N(p,q)$ for integers $p,q$.} We have a single source $s$. For each edge $e$ in the graph balancing instance, we have a node $n_e$ in the network. There is an arc from $s$ to each such $n_e$ 
with flow {\em lower bound} of $1$ if $e$ is a small edge, and a flow lower bound of $p$ otherwise. We have two nodes, $v_{i,b}$ and $m_i$, for each vertex $i$ in the graph balancing instance.
There is an arc of flow upper bound $p$ from $v_{i,b}$ to $m_i$.
For every big edge $e$, there is an arc from $n_e$ to $v_{i,b}$ iff $e$ is adjacent to $i$ with flow upper bound $p$. 
For every small edge $e$, there is an arc from $n_e$ to $m_{i}$ iff $e$ is adjacent to $i$ with flow upper bound $1$. 
Therefore, every $n_e$ has out-degree $2$ unless it is a loop in which case it has out-degree $1$. 
Finally, we have one sink node $t$ and each node $m_i$ has an arc to $t$ with flow upper bound  $q$. 
%
 An illustration of the flow network is shown in Figure~\ref{fig:gb}.
\begin{figure}[h!]
	\begin{center}
		\includegraphics[scale=0.3]{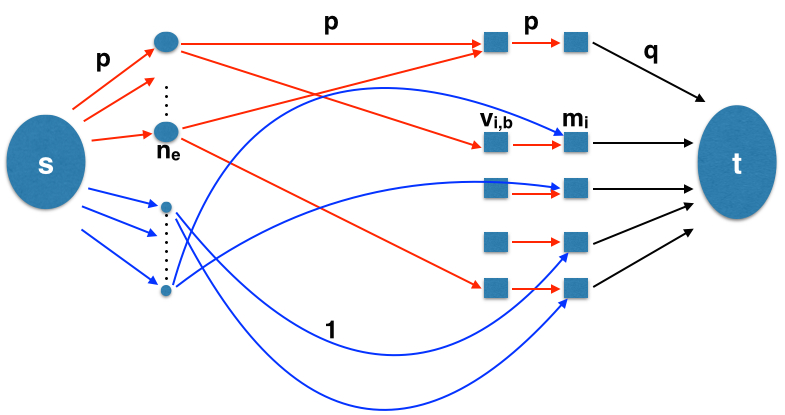}
	\end{center}
	\caption{\small Big circles denote big edges and small ones denote small edges. All the red arcs have weight $p$, blue arcs have weight $1$, and black arcs have weight $q$}
	\label{fig:gb}
\end{figure}

\begin{lemma}\label{lem:21}
If $\OPT = 1+\kk c < 2$ for some integer $\kk$, then there is a feasible flow in the network $N(k,\kk+k)$. 
\end{lemma}
\begin{proof}
	For small edges $e$, we send $1$ unit of flow from $n_e$ to $m_i$ where $i$ is the vertex towards which $e$ is oriented. 
	For every big edge $e$, we send $k$ units of flow from $n_e$ to $v_{i,b}$ to $m_i$, where $i$ is the vertex towards which $e$ is oriented. 
	Finally $m_i$ sends all the flow it receives to $t$.

	Since $\OPT = 1+\kk c < 2$, in the optimal allocation each vertex has at most one big edge oriented towards it, and so no $v_{i,b}$ receives flow from two separate $n_e$'s. So the flow on $(v_{i,b},m_i)$ edge is at most $k$. If a machine $m_i$ receives $k$ units of flow from $v_{i,b}$, then it must be that in the graph balancing solution $i$ has one big edge oriented towards it and since $\OPT = 1 + \kk c$ is has at most $\kk$ small edges are oriented towards $i$. Thus, the total flow that $m_i$ receives is at most $(k+\kk)$. If machine $m_i$ receives $0$ units of flow from $v_{i,b}$, then in the graph balancing solution there are at most $\floor{\frac{1+\kk c}{c}} = \kk + \floor{1/c} = \kk+k$ small edges oriented towards $i$, implying the total flow from $m_i$ to $t$ is at most $(k+\kk)$.
\end{proof}

\begin{lemma}\label{lem:22}
If $\OPT = \kkkk c < 2$ for some integer $\kkkk$ and not equal to $1+\kk c$ for any non-negative integer $\kk$, there is a feasible flow in the network $N(k+1,\kkkk)$. 
\end{lemma}
\begin{proof}
	Since $\kkkk c \neq 1 + \kk c$ for any integer $\kk$, we get $1/c$ is not an integer for otherwise $\kkkk c = 1 + c\cdot\left( \kkkk - 1/c\right)$. Therefore, $\ceil{1/c} = k+1$.
	We will use this in the proof.
	
	For small edges $e$, we send $1$ unit of flow from $n_e$ to $m_i$ where $i$ is the vertex towards which $e$ is oriented. 
	For every big edge $e$, we send $k+1$ units of flow from $n_e$ to $v_{i,b}$ to $m_i$, where $i$ is the vertex towards which $e$ is oriented. 
	Finally $m_i$ sends all the flow it receives to $t$.
	
	Since $\OPT = \kkkk c < 2$, in the optimal allocation each vertex has at most one big edge oriented towards it, and so no $v_{i,b}$ receives flow from two separate $n_e$'s. So the flow on $(v_{i,b},m_i)$ edge is at most $k+1$. If a machine $m_i$ receives $k+1$ units of flow from $v_{i,b}$, then it must be that in the graph balancing solution $i$ has one big edge oriented towards it.
	Therefore the number of small edges oriented towards $i$ is at most  $\floor{\frac{\kkkk c - 1}{c}} = \floor{\kkkk - 1/c} = \kkkk - \ceil{1/c} = \kkkk - (k+1)$.
		 Thus, the total flow that $m_i$ receives is at most $\kkkk$. If machine $m_i$ receives $0$ units of flow from $v_{i,b}$, then in the graph balancing solution there are at most $\kkkk$ small edges oriented towards $i$, implying the total flow from $m_i$ to $t$ is at most $\kkkk$.
\end{proof}
\noindent
Since all capacities in $N(p,q)$ are integral, by integrality of flows,  if there is a feasible fractional flow in $N(p,q)$, then there must be a feasible integral flow in $N(p,q)$.

\begin{lemma}\label{lem:3}
	Given an integral feasible flow in $N(p,q)$, we can obtain a schedule with makespan $\leq \max \{c q, 1 + c\cdot (q - \floor{\frac{p+1}{2}})\}$.
\end{lemma}
\begin{proof}
	For any small edge $e$, we know that $n_e$ receives one unit of flow from $s$, and using the integrality of flows, $n_e$ sends one unit of flow to exactly one machine $m_i$. We orient $e$ towards vertex $i$.
    For any large edge $e$, if the flow on any of the out-going arcs, say $(n_e,v_{i,b})$,  is {\em strictly} greater than $\floor{p/2}$, then orient $e$ towards $i$.  This defines a partial orientation of all the edges of the graph balancing instance. \\
    
   \noindent
    {\bf Case 1: $p$ is odd.} Let $p=2\ell + 1$. Since a node $n_e$ corresponding to a big edge $e$ receives a flow of $2\ell + 1 $ units from $s$, and there are at most two outgoing arcs from $n_e$, 
    by integrality of flows and pigeonhole principle, it must be the case that one of the two arcs  carries at least $\ell+1 > \floor{p/2}$ flow. In particular, all big edges are oriented in the partial orientation implying it is in fact a complete one.
   
   Let us now analyze the total weighted in-degree of any vertex $i$.  Firstly note that if a vertex $i$ only has small edges oriented towards it, then its weighted in-degree is $\leq c q$. 
   Secondly, any vertex $i$ can have at most one big-edge oriented towards it. To see this, note that if $i$ got two big edges $e$ and $e'$ oriented towards it, then the flow on $(n_e,v_{i,b})$ and $(n_{e'},v_{i,b})$ are both $\geq \ell + 1$. All this flow of $\geq 2\ell + 2$ units must flow through the $(v_{i,b},m_i)$ arc, whose upper bound is $p = 2\ell+1$ contradicting the feasibility of the flow. So each vertex $i$ has at most one big-edge oriented towards it.  Furthermore, if this is the case, then the flow on the $(v_{i,b},n_e)$ arc is $\geq \ell + 1$ implying the total flow $i$ obtains from arcs of the form $(n_{e'},m_i)$ for small edges $e'$ is at most $q - (\ell + 1)=q - (\frac{p+1}{2})$. Thus the total weighted indegree of $i$ is at most $1 + c\cdot (q - \frac{p+1}{2})$.\\ 
   	
\noindent
	{\bf Case 2: $p$ is even.} Let $p=2\ell$. In this case, the current orientation may indeed be partial.
	Let $F$ be the collection of nodes $n_e$ which correspond to unoriented big edges, and $U$ be the out-neighbors of $F$ in the flow network. 
	Let $e$ be one such edge which hasn't been oriented yet. It must be the case then that the flow of $2\ell$ units that $n_e$ receives from $s$ must flow out 
	on {\em exactly} two arcs $(n_e,v_{i,b})$ and $(n_e,v_{j,b})$, each carrying flow exactly $\ell$ (in particular, $e$ cannot be a loop).
	Furthermore, if $v_{i,b}\in U$ receives $\ell$ units of flow from some $n_e \in F$, then since the capacity of $(v_{i,b},m_i)$ arc is $2\ell$ we have (a) it receives flow from at most one other vertex in $F$,
	and (b)	it cannot receive $\ell+1$ units of flow from any other $n_{e'}$; in particular, the vertex $i$ doesn't have any big edge oriented towards it in the partial orientation.
	
	Now in the induced directed graph between $F$ and $U$, the above discussion implies that the out-degree of every vertex in $F$ is exactly $2$, and the in-degree of every vertex in $U$ is at most $2$.
	This implies there is a matching between $F$ and $U$ which completely matches $F$; one can easily check Hall's condition. If $n_e$ is matched to $v_{i,b}$, then we orient $e$ towards $i$, thereby extending the partial orientation to a complete one. Note that any node $i$ has at most one edge corresponding to $F$ oriented towards it.
	
	Let us now analyze the weighted in-degree of a vertex $i$. Once again, as in Case 1, if $i$ has only small edges oriented towards it, then its weighted indegree is at most $c q$.
	Also as in Case 1, any vertex $i$ can have at most one big edge oriented towards it in the partial orientation. 
	Furthermore, by point (b) above, if $i$ receives a big edge in the partial orientation, it doesn't receive any other big edges in the extension.
    On the other hand if $i$ doesn't receive a big edge in the partial orientation, then by the matching property of the extension, it receives at most one big edge in the extension.
    Finally, note that if $i$ does receive exactly one big edge, then the total flow on the $(v_{i,b},m_i)$ arc must be $\geq \ell$. So, the total flow on arcs $(n_e,m_i)$ for small edges $e$ is at most $q - \ell = q-p/2$. Therefore, the weighted in-degree of a vertex $i$ which receives at lease one big edge is at most $1 + c(q-p/2)$.
\end{proof}

\begin{lemma}\label{lem:41}
	Given $\OPT = 1+\kk c < 2$, we can obtain a schedule with makespan $\leq \frac{3}{2} \OPT$.
\end{lemma}
\begin{proof}
By Lemma~\ref{lem:21} and Lemma~\ref{lem:3}, there is a feasible flow in the network $N(k,\kk+k)$ and we can obtain a schedule with makespan $\leq \max \{c \cdot (\kk+k), 1 + c \cdot(\kk+k-\floor{\frac{k+1}{2}})\}$.
Now, since $k = \floor{1/c}$, we have $kc < 1$ implying $c(\kk + k) < 1 + \kk c = \OPT$.
Also, we get $k - \floor{\frac{k+1}{2}} =  \floor{\frac{k}{2}}$, and so we get that the $\max \{c \cdot (\kk+k), 1 + c \cdot(\kk+k-\floor{\frac{k+1}{2}})\} \leq \OPT + c\cdot\floor{k/2} \leq \OPT  + 1/2 < 3\OPT/2$.
\end{proof}

\begin{lemma}\label{lem:42}
	Given $\OPT = \kkkk c < 2$ and is not equal to $1+\kk c$ for any non negative integer $\kk$, we can obtain a schedule with makespan $\leq \frac{3}{2} \OPT$.
\end{lemma}
\begin{proof}
By Lemma~\ref{lem:22} and Lemma~\ref{lem:3}, there is a feasible flow in the network $N(k+1,\kkkk)$ and we can obtain a schedule with makespan $\leq \max \{c \cdot \kkkk, 1 + c \cdot(\kkkk-\floor{\frac{k+2}{2}})\} = \OPT + 1 - c \cdot (\floor{\frac{k}{2}}+1) =\OPT + 1 - c \cdot \ceil{\frac{k+1}{2}}  \leq \OPT + 1/2 \leq \frac{3}{2}\OPT$ (Since $c \cdot \ceil{\frac{k+1}{2}} \geq c \cdot \frac{k+1}{2} > 1/2$)
\end{proof}
The above two lemmas along with Lemma~\ref{lem:1} implies the following theorem.
\begin{theorem}\label{thm:main}
	There is a $3/2$-approximation for graph balancing with two types of edge weights.
\end{theorem}

\noindent
{\bf Remark.} The above theorem can also show that the integrality gap of the so-called configuration LP is $\leq 3/2$ too. To see this, note it is easy to prove that if the conf LP is feasible for some guess $T = 1+\kkk c$ or $k''c$ of $\OPT$, 
then there is a feasible flow in the network $N(k,\kk + k)$ or $N(k+1,\kkkk)$, respectively.\smallskip

\noindent
{\bf Remark.}  (added 9th June, 2016.) After posting the version 1 of our paper, it was brought to our notice that our result is not new. Huang and Ott~\cite{HO15}, and independently, Page and Solis-Oba~\cite{PS16} also obtain the same result; the former paper in fact gives a 3/2-factor algorithm even in the case when the smaller jobs can go to any number of machines. It is easy to see that our algorithm easily modifies to this case (we never use the fact that small jobs go to at most two machines). The algorithm and analysis in this note is arguably simpler.

\bibliographystyle{alpha}
\bibliography{GBpapers}
\end{document}